\pdfoutput=1
\documentclass[conference]{IEEEtran}

\ifCLASSINFOpdf

\else

\fi

\usepackage{amsmath}
\usepackage{amsthm}

\newtheorem{lemma}{Lemma}
\newtheorem{remark}{Remark}

\newtheorem{auxiliary code}{Auxiliary Code}

\newtheorem{definition}{Definition}

\newtheorem{strategy}{Sampling Strategy}
\usepackage{algorithm}
\usepackage[noend]{algpseudocode}

\algrenewcommand\algorithmicindent{0.6em}%

\usepackage{tikz}
\usepackage{graphicx}
\usetikzlibrary{positioning}
\usepackage[font={small}]{caption}
\usepackage{subcaption}

\def \tS {\mathcal{S}}
\def \tC {\mathcal{C}}
\def \td {\zeta}
\def \tH {\mathcal{H}}
\def \tQ {Q}
\def \ttc {\lambda}
\def \tcn {c}
\def \ttg {g_c}
\def \tGEC {\mathcal{\widetilde{G}}}
\def \teta {\mu}
\def \tK {\mathcal{K}}
\def \talpha {\alpha_{\ttg}}
\def \tV {\mathcal{V}}

\usepackage{multicol}

\newcommand\lev[1]{{\color{black}#1}}
\newcommand\deb[1]{{\color{black}#1}}

\newcommand\redtext[1]{{\color{black}#1}}
\newcommand\bluetext[1]{{\color{black}#1}}
\newcommand\greentext[1]{{\color{black}#1}}

\begin{document}

\title{ Concentrated Stopping Set Design for \\
Coded Merkle Tree: Improving Security Against Data Availability Attacks in Blockchain Systems}

\author{\IEEEauthorblockN{Debarnab Mitra, Lev Tauz and Lara Dolecek}
\IEEEauthorblockA{Department of Electrical and Computer Engineering, University of California, Los Angeles, USA\\
email: debarnabucla@ucla.edu, levtauz@ucla.edu, dolecek@ee.ucla.edu}
\vspace{-0.9cm}}

\maketitle

\begin{abstract}
In certain blockchain systems, light nodes are clients that download only a small portion of the block. Light nodes are vulnerable to \emph{data availability} (DA) attacks where a malicious node hides an invalid portion of the block from the light nodes. \lev{Recently, a technique based on erasure codes called Coded Merkle Tree (CMT) was \deb{proposed by Yu \textit{et al.}} that enables light nodes to detect a DA attack with high probability. The CMT \deb{is constructed using} LDPC codes for fast decoding but can fail to detect a DA attack if a malicious node hides a small stopping set of the code.} To combat this, \deb{Yu \textit{et al.}} used well-studied techniques to design random LDPC codes with high minimum stopping set size. Although effective, these codes are not necessarily optimal for this application.  In  this paper, we demonstrate a more specialized LDPC code design to improve the security against DA attacks. We achieve this goal by providing a deterministic LDPC code construction that focuses on concentrating stopping sets to a small group of variable nodes rather than only eliminating stopping sets. We design these codes by modifying the Progressive Edge Growth algorithm into a technique called the \emph{entropy-constrained} PEG (EC-PEG) algorithm. This new method demonstrates a higher probability of detecting DA attacks and allows for good codes at short lengths.
\end{abstract}

\IEEEpeerreviewmaketitle

   \vspace{-10pt}
\section{Introduction}
   \vspace{-2pt}
Blockchain systems typically have two types of nodes: \textit{i) Full nodes -} they store the full blockchain ledger and can validate transactions by operating on the entire ledger, \textit{ii) Light nodes -} they do not store the entire ledger and hence cannot validate transactions. 
Light nodes are implemented using the Simplified Payment Verification (SPV) technique \cite{Bitcoin}: For each block, a Merkle tree is constructed using the transactions in the block as leaf nodes, and the Merkle root is stored at the light nodes.
Using the Merkle root, light nodes can verify the inclusion of transactions in a block via a Merkle proof \cite{Bitcoin} but not its correctness.

Networks in which light nodes are connected to a majority of malicious full nodes are susceptible to \emph{data availability (DA)} attacks
\cite{CMT}, \cite{dataAvailOrg} where a malicious full node generates a block with invalid transactions, publishes a Merkle root that satisfies the Merkle proof (thus allowing light nodes to successfully verify the inclusion of the transactions in the block), and hides the invalid portion of the block so that honest full nodes are unable to validate the block and notify the light nodes of the malicious behaviour. Light nodes can independently detect a DA attack if an anonymous request for a portion of the block is rejected by the full node that generates the block.  As such, light nodes can randomly sample the block, i.e., randomly request for different portions of the block transactions, to detect a DA attack. \deb{The detection, however,} becomes increasingly unlikely as the block size increases since a malicious node can hide a very small section. To alleviate this problem, authors in \cite{dataAvailOrg} \lev{introduced the method of coding the block using erasure codes such that it forces the adversary to hide a larger portion of the block which, then, can be detected with a high probability by the light nodes using random sampling}. Work in \cite{CMT} extends the idea into a technique called Coded Merkle Tree (CMT) that uses Low-Density Parity-Check (LDPC) codes to encode each layer of the Merkle tree. The advantage of an LDPC code is that it enables the use of a polynomial time peeling decoder at the full nodes to decode the coded symbols. 

A stopping set of an LDPC code is a set of variable nodes (VNs) such that every check node (CN) connected to this set is connected to it at least twice. The peeling decoder fails to fully decode a layer of the CMT if the malicious node hides coded symbols corresponding to a stopping set of the LDPC code used in the layer. 
Since the malicious node can hide the smallest stopping set of the LDPC code,
the probability of failure (at a particular layer) for the light nodes to detect a DA attack using random sampling depends on the stopping ratio (size of the smallest stopping set divided by the codeword length \lev{\cite{CMT}}) of the LDPC code. 

The CMT can be constructed using any LDPC code. In \cite{CMT}, authors used codes from a well studied random LDPC ensemble that guarantees a certain stopping ratio with high probability. Despite the high stopping set ratio, these codes were originally designed for other types of channels (i.e., BSC) and are not necessarily optimal for this specific application. In this paper, we provide a more specialized LDPC code design to \lev{significantly lower} the probability of light node failure. We accomplish this \lev{goal} by providing a deterministic LDPC code construction that focuses on concentrating stopping sets to a small set of VNs based on the Progressive Edge Growth algorithm \cite{PEG} which we term as the \emph{entropy-constrained} PEG (EC-PEG) algorithm. Intuitively, a greedy sampling strategy can now be used by the light nodes where they sample this small set of VNs to guarantee that a large number of stopping sets are being sampled. We shall demonstrate that this novel construction provides a much lower probability of light nodes failing to detect a DA attack than for an LDPC code chosen from a random ensemble. Additionally, using a deterministic algorithm, we can design LDPC codes for small to moderate blocklengths where  random ensembles \lev{typically} have trouble providing good codes. 
 The rest of this paper is organized as follows: In Section \ref{sec:system_model}, we provide  preliminaries and describe the system model. In Section \ref{sec:idea}, we provide motivation for our design method. The description of the EC-PEG algorithm is provided in Section \ref{sec:EC_PEG}. Finally, we present our simulation results in Section \ref{sec:sims} and concluding remarks in Section \ref{sec:conclusion}. 

\section{Preliminaries and System Model}\label{sec:system_model}

Similar to the Merkle tree introduced in \cite{Bitcoin}, a CMT is constructed by applying an LDPC code to each layer of the Merkle tree before hashing the layer to generate its parent layer. While we focus on the relevant issues of the CMT, specific details regarding the CMT can  be found in \cite{CMT}. 
 We consider a blockchain network with full nodes and light nodes,
 where only full nodes can produce new blocks. Light nodes do not store the entire block but store the root of \lev{the} CMT and use it to verify block inclusions via Merkle proofs \cite{CMT}. We also assume that each light node is honest and is connected to at least one honest full node. However, it may be connected to a majority of malicious full nodes. 
 
 In this paper, we consider one layer of the CMT and assume that it has $n$ coded symbols generated using an \deb{LDPC code with a binary parity check matrix $H$ (having $n$ columns)}. Techniques developed for this layer can be applied to all layers. 
 \lev{A DA attack occurs when a malicious node 1) produces a block that satisfies the Merkle proof for the light nodes to accept it as a valid block, and 2) hides a small portion of the coded symbols, corresponding to a stopping set of $H$, such that honest full nodes are not able to decode the block fully using a hash aware peeling decoder \cite{CMT}}. The goal of the CMT is to prevent light nodes from accepting a block whose CMT layer cannot be decoded by a full node. 
 To accomplish this goal, light nodes anonymously request coded symbols from the block producer according to some sampling strategy and accept the block if all the requested symbols are returned.
 The malicious producer only returns coded symbols that are not hidden from the system. \deb{Thus, a light node fails to detect a DA attack if the samples requested are not hidden.}

In this paper, we are concerned with the probability of a light node failing to detect a DA attack.
Let $p_f(s,\teta)$ be the probability of failure for a 
given $H$
and some sampling strategy when the light node uses $s$ samples and the malicious node hides a stopping set with $\teta$ VNs. 
In the CMT in \cite{CMT}, random sampling with replacement is employed and $p_f(s,\teta) = (1 - \frac{\teta}{n})^s$. For a given $s$, $p_f(s,\teta)$ is maximized if the malicious node hides the smallest stopping set of \deb{$H$}.
In \cite{CMT}, \lev{the LDPC code} is chosen from a random code ensemble that, with high probability, has a stopping ratio $\beta^*$. Thus, the probability of failure becomes $(1 - \beta^*)^s$, w.h.p..
 The system is less secure when the stopping ratio of \lev{$H$} is less than $\beta^*$, which happens with a large probability at short block lengths. To improve security, 
 we seek to provide a deterministic LDPC code construction and an associated sampling strategy that still results in a low probability of failure $p_f(s,\teta)$ for small $\teta$. 
 For the rest of the paper, we assume that the malicious node is unaware of the sampling strategy used by the light nodes and for a given size of the stopping set $\teta$, it hides a randomly chosen stopping set of size $\teta$.

  \vspace{-0.1cm}
 \begin{remark}\label{remark:comparisons}
 \redtext{In this work, compared to \cite{CMT}, we focus on a new design of the constituent LDPC code used for the construction of the CMT. As a result, performance metrics and their order-optimal solutions provided in \cite{CMT} are not compromised. In particular: i) the header remains as the CMT root and its size does not grow with the blocklength; ii) the decoding complexity of the hash-aware peeling decoder is still linear in the blocklength; iii) the incorrect coding proof size for our codes is empirically shown to be similar to codes from \cite{CMT}.}
 \end{remark}
 
 \vspace{-0.1cm}
 For a parity check matrix $H$ with $n$ columns $\tV = \{v_1, v_2,\ldots, v_{n}\}$, let $\mathcal{G}$ denote the  Tanner graph (TG) representation  of $H$.
 We  \lev{interchangeably} refer to $v_i$ as the $i^{th}$ VN in $\mathcal{G}$ and rows of $H$ as CNs in $\mathcal{G}$.  For a set $\mathcal{S}$, let $|\mathcal{S}|$ denote its cardinality.
\lev{Let $g_{\min}$ denote the girth of $\mathcal{G}$}.
 $H_{\tQ}$ and $\overline{H}_{\tQ}$, $\tQ \subseteq \tV$, represent submatrices of $H$ with columns $\{v_i | v_i \in \tQ\}$ and $\{v_i | v_i \in \tV\setminus \tQ\}$, respectively, and all non-zero rows. We say that $H_{\tQ}$ has a cycle of length $\ell$ (\bluetext{called an $\ell$-cycle}) iff the induced subgraph of $\mathcal{G}$ by the VNs \deb{in} $\tQ$ has an \bluetext{$\ell$-cycle}.  For a cycle (stopping set) in the $\mathcal{G}$, we say that a VN $v$ \emph{touches} the cycle (stopping set) iff $v$ is part of the cycle (stopping set). \lev{The set} $\tS  = \{v_{i_1}, v_{i_2}, \ldots, v_{i_s}\}
 \subseteq \tV$ is a sample set if light nodes sample the VNs $\{v_{i_1}, v_{i_2}, \ldots, v_{i_s}\}$.
 We define the weight of a stopping set as the no. of VNs touching it. A stopping set is hidden by a malicious node if all VNs present in it are hidden. 
 For any stopping set of weight $\teta$ hidden by the malicious node, let $P_f(\tS , \teta)$ be the failure probability of light nodes with the sample set $\tS$. For $\mathrm{p} = (p_1, p_2 \ldots, p_t)$ such that $p_i \geq 0$,   $\sum_{i=1}^{t}p_i = 1$, we define the entropy function $\mathcal{H}(\mathrm{p}) = \sum_{i=1}^{t}p_i\log(p_i)$ (assume that $0\log 0 = 0$).

\section{Design Idea: Stopping set Concentration}\label{sec:idea}

In this section, we motivate our novel design idea of concentrating stopping sets to improve light node detection of DA attacks. 
It was shown  in \cite{WesselSS} that for LDPC codes with no degree one VNs, stopping sets are made up of cycles. Since working with stopping sets directly is generally computationally difficult, we focus on concentrating cycles. It is known that codes with irregular VN degree distributions are prone to small stopping sets. Thus, we consider VN degree regular LDPC codes of VN degree $d_v \geq 3$. We call such codes \emph{$d_v$-reg} LDPC codes. 
  
\vspace{-0.1cm}
 \begin{definition}
 \label{remark:sigmabound}
 Let $\Sigma(d_v,g)$ denote the minimum weight possible for a stopping set that only has cycles of length $\geq g$ and all VNs have degree $d_v$. Thus, in a $d_v$-reg LDPC code, all stopping sets of weight $< \Sigma(d_v,g)$ have at least one cycle of length $< g$.
 \end{definition}

\vspace{-0.1cm}
 Closed form expressions for $\Sigma(d_v,6)$, $\Sigma(d_v,8)$ and a lower bound for $\Sigma(d_v,g)$, $g\geq 10$, \lev{are} provided in \cite{sigmadg,baniSS}, e.g., $\Sigma(d_v,10) \geq 1 + d^2_v$. 
 Next, we consider the following definition of a sample set followed by an immediate lemma:
 
\vspace{-0.1cm}
\begin{definition}\label{defn:sampleset}
A $g$-sample set $\tS_g$ is a sample set such that $\overline{H}_{\tS_g}$ has no cycles of length  $< g$.
\end{definition}
 
\vspace{-0.15cm}
 \begin{lemma}\label{lemma:zeroprob}
 For a $d_v$-reg LDPC code, if light nodes sample a $g$-sample set $\tS_g$, then for all stopping sets of weight $\teta \leq \Sigma(d_v,g)-1$, $P_f(\tS_g,\teta) = 0.$
 \end{lemma}
 \begin{proof}  \lev{The result follows from Definitions \ref{remark:sigmabound} and \ref{defn:sampleset}}.
 \end{proof}
 
\vspace{-0.15cm}
From the above lemma, if the sample set of the light node includes a $g$-sample set, it  fails only if the malicious node hides a stopping set of weight $\geq \Sigma(d_v,g)$ that can be caught using further random sampling.
Clearly, we wish to use the smallest $g$-sample set to improve efficiency, but \lev{this \deb{smallest} set is} computationally hard to determine. Instead, we use a greedy method to find a $g$-sample set $\tS^a_g$ of small \lev{enough} size that is shown to perform well.

Let $\tC_g$ denote a set that touches all \bluetext{$g$-cycles} obtained using the following greedy method. For $g=g_{\min}$, initialize $\tC_{g_{\min}} = \emptyset$, otherwise $\tC_{g} = \cup_{g_{\min} \leq g' < g}\; \tC_{g'}$. Purge all the VNs in $\tC_g$ from $\mathcal{G}$. \greentext{Then, select a VN $v$ at random from the VNs that touch the maximum number of \bluetext{$g$-cycles} in $\mathcal{G}$. Add $v$ to $\tC_g$, purge $v$ from $\mathcal{G}$, and repeat the process until no \bluetext{$g$-cycles} are left in $\mathcal{G}$. Our greedy $g$-sample set is  $\tS^a_g = \tC_{g-2}.$} \bluetext{A formal algorithm for the above can be found in Appendix \ref{appendix:alg:g-sample}.}

Our goal is to design \deb{parity check matrices $H$}  that have small $|\tS^a_g|$. \bluetext{Let $\td^g = (\td^g_1, \td^g_2, \ldots, \td^g_n)$ be the VN-to-$g$-cycle distribution} where $\td^g_i$ is the fraction of \bluetext{$g$-cycles} touched by $v_i$. Due to the construction of $\tS^a_g$, $|\tS^a_g|$ is small if a majority of \bluetext{$g'$-cycles, $g' < g$,} are touched by the same small subset of VNs. This goal is achieved if for all \bluetext{$g'$-cycles}, $g' < g$, distributions $\td^{g'}$ are concentrated (i.e., have a high \bluetext{$g'$-cycle} fraction) towards the same set of VNs.
Later, we design the EC-PEG algorithm that \greentext{achieves concentrated distributions}.

\vspace{-0.05cm}
\begin{definition}\label{remark:skewingSS}
Let $ss^{\teta} = (ss^{\teta}_1, ss^{\teta}_2, \ldots, ss^{\teta}_n)$ be the \lev{VN-to-stopping set} of weight $\teta$ distribution where $ss^{\teta}_i$ is the fraction of stopping sets of weight  $\teta$ touched by $v_i$. Due to Definition \ref{remark:sigmabound}, by concentrating the distributions $\td^{g'}$ for $g' < \ttg$, for some cycle length $\ttg$, we also concentrate $ss^{\teta}$, where $\teta \leq \Sigma(d_v,\ttg)-1$.
\end{definition}
\vspace{-0.1cm}
 The next lemma demonstrates that LDPC codes with concentrated $ss^{\teta}$ results in a smaller probability of light node failure when a malicious node randomly hides a stopping set of size $\teta$ and the light nodes use a sample set of size $s$.

\vspace{-0.05cm}
\begin{lemma}\label{lemma:average_pf}
Let $\mathcal{SS}_{\teta}$ denote the set of all weight $\teta$ stopping sets in the LDPC code. If the malicious node picks a stopping set randomly from  $\mathcal{SS}_{\teta}$ and hides it, the probability of failure at the light node $p_f(s,\teta)$ using $s$ samples and any sampling strategy satisfies $p_f(s,\teta) \geq 1 - \max_{\tS \subseteq \tV, |\tS| = s}\tau(\tS,\teta)$,
where $\tau(\tS,\teta)$ is the fraction of stopping sets of weight $\teta$ touched by the sample set $\tS$. Let $\tS^{opt}_{\teta} =  \mathrm{argmax}_{\tS \subseteq \tV, |\tS| = s}\tau(\tS,\teta)$. The lower bound in the above equation is achieved when light nodes sample, with probability one, the set $\tS^{opt}_{\teta}$.\vspace{-0.05cm}
\end{lemma}
\begin{proof}
Proof can be found in Appendix \ref{appendix:lemma:average_pf}.
\end{proof}

\vspace{-0.15cm}
According to the above lemma, for a fixed sampling size of $s$, the lowest probability of failure is achieved when the light nodes sample the set $\tS^{opt}_{\teta}$ and hence, the optimal probability of failure becomes $1 - \tau(\tS^{opt}_{\teta},\teta)$. Thus, from Definition \ref{remark:skewingSS}, by designing LDPC codes that have  concentrated $ss^{\teta}$, we increase the value of $\tau(\tS^{opt}_{\teta},\teta)$ and reduce the probability of failure. 

\vspace{-0.05cm}
\begin{strategy}\label{remark:samplingstrategy}
\lev{We are unaware of an efficient method to find $\tS^{opt}_{\teta}$}, thus we construct a sample set using greedy techniques.  Since we aim to concentrate the distributions  $\td^{g'}$ for $g_{\min} \leq g' < \ttg$ towards the same VNs, we form a greedy sample set $\tS_{greedy}$ by greedily choosing $s$ samples from $\tS^a_{\ttg}$ (assuming $s \leq |\tS^a_{\ttg}|$). \lev{Assume $C_{\tilde{g}} = \emptyset$ for $\tilde{g} < g_{\min}$. Find} the smallest $\tilde{g}$ that satisfies $|\tC_{\tilde{g}}| > s$. Then, select $s-|\tC_{\tilde{g}-2}|$ samples from $\tC_{\tilde{g}} \setminus \tC_{\tilde{g}-2}$ that touch the most \bluetext{$\tilde{g}$-cycles}, which we define as $\widetilde{\tC}$. As such, $\tS_{greedy} = \tC_{\tilde{g}-2} \cup \widetilde{\tC}$. Our sampling strategy is sampling with probability one the set $\tS_{greedy}$. Thus, the same $s$ samples are used irrespective of the stopping set size $\teta$. The probability of failure using this strategy is $p_f(s,\teta) = 1 - \tau(\tS_{greedy},\teta)$ (see proof of Lemma \ref{lemma:average_pf}).
\end{strategy}

\vspace{-0.1cm}
We verify from simulations that by concentrating the distributions $\td^{g'}$ for $g_{\min} \leq g' < \ttg$, and using sampling strategy \ref{remark:samplingstrategy}, we get a failure probability $p_f(s,\teta)$ much lower than is achievable using the techniques in \cite{CMT} for short block lengths. 
Next, we provide our method to design LDPC codes that result in concentrated distributions $\td^{g'}$ for $g_{\min} \leq g' < \ttg$, for some upper bound cycle length $\ttg$. Choice of $\ttg$ is a complexity constraint that determines how many cycles we keep track of in the EC-PEG algorithm.

\vspace{-0.05cm}
\section{Entropy-constrained PEG algorithm}\label{sec:EC_PEG}
 Algorithm \ref{alg:EC-PEG} presents our EC-PEG algorithm
 for constructing a TG $\tGEC$ with $n$ VNs, $m$ CNs, and VN degree $d_v$  that concentrates the distributions $\td^{g'}$ for all \bluetext{$g'$-cycles} with $g'< \ttg$. 
 
\vspace{-0.1cm}
 \begin{algorithm}
\caption{EC-PEG Algorithm}\label{alg:EC-PEG}
\begin{algorithmic}[1]
\State \textbf{Inputs:} $n$, $m$ , $d_v$, $\ttg$  
\State \textbf{Initialize} $\tGEC$ to $n$ VNs, $m$ CNs and no edges
\State \textbf{Initialize} $\ttc^{g'}_i = 0$, for all $g'  < \ttg$ and $1\leq i \leq n$
\For{$j=1$ to $n$}
    \For{$k=1$ to $d_{v}$}
        \If{$k=1$}
        \State
    $\tcn^{s}$ = ``first CN selection" procedure
    \Else{ [$\tK , g', F$] = $PEG(\tGEC, v_j)$}
        \If{ $F = 0$ or ($F = 1$ and $g' \geq \ttg$)} 
        \State
    $\tcn^{s}$ = ``min degree CN selection”  procedure
        \Else\Comment{\textit{(\bluetext{$g'$-cycles}, $g' < \ttg$, are created)}}
        \State
            $\tcn^{s}$ = ``entropy-constrained CN selection” procedure
          \State 
          Update $\ttc^{g'}$, $\alpha^{g'}$ and $\talpha$  
        \EndIf
    \EndIf    
    \State  $\tGEC = \tGEC \cup \mathrm{edge}\{\tcn^{s},v_j\}$

    \EndFor
\EndFor
\State \textbf{Outputs:} $\tGEC$, $g_{\min}$
\end{algorithmic}
\end{algorithm}

 \vspace{-0.15cm}
  The PEG algorithm builds a TG in an edge by edge manner by iterating over the set of VNs and for each VN $v_j$, establishing $d_{v}$ edges to it.
  Similar to the original PEG algorithm proposed in \cite{PEG}, the ``first CN selection" procedure is just selecting a CN which has the minimum degree under the current form of the TG $\tGEC$. If multiple CNs exist with the same minimum degree, a CN is selected randomly.
  
  When establishing the $k^{th}$ edge to VN $v_j$, $k \geq 2$, the PEG
  
  \noindent
  algorithm encounters two situations: i) addition of the edge is possible without creating cycles; ii) addition of the edge creates cycles. In both the situations, the PEG algorithm finds a set of \emph{candidate} CNs that it proposes to connect $v_j$ to, in order to maximize the girth.
Without getting into the details of the PEG algorithm described in \cite{PEG}, we assume a procedure $PEG(\mathcal{G},v_j)$ that provides us with the set of candidate CNs $\tK$ for establishing a new edge to VN $v_j$ under the TG setting $\mathcal{G}$ according to the PEG algorithm described in \cite{PEG}. We also assume that the procedure returns flags $F = 0$  and $F = 1$ respectively for situations i) and ii) described above. For ii), it also returns the cycle length $g'$ of the smallest cycles formed when an edge is established between any CN in $\tK$ and $v_j$. 
For the case of $k \geq 2$, we call the procedure $PEG(\tGEC, v_j)$. When $F = 0$, the CN candidates in $\tK$ do not create any cycles and similar to the original PEG algorithm, a CN is selected from $\tK$ with the minimum degree under the current TG setting $\tGEC$. If multiple CNs exist with the minimum degree, a CN is picked randomly. This is the ``min degree CN selection" procedure. 

When $F = 1$, \bluetext{$g'$-cycles} are created when a new edge is established between any CN candidate in $\tK$ and $v_j$. While progressing through the EC-PEG algorithm, for all \bluetext{$g'$-cycles}, $g' < \ttg$, we maintain \bluetext{\textit{VN-$g'$-cycle}} counts $\ttc^{g'} = (\ttc^{g'}_1, \ttc^{g'}_2, \ldots, \ttc^{g'}_n)$, where $\ttc^{g'}_i$ is the number of \bluetext{$g'$-cycles} that are touched by VN $v_i$. We also maintain \bluetext{\textit{VN-$g'$-cycle}} normalized counts $\alpha^{g'} = (\alpha^{g'}_1, \alpha^{g'}_2, \ldots, \alpha^{g'}_n)$, where $\alpha^{g'}_i = \frac{\ttc^{g'}_i}{\sum_{i=1}^{n}\ttc^{g'}_i}$ if $\sum_{i=1}^{n}\ttc^{g'}_i \neq 0$, else $\alpha^{g'}_i = 0$, and joint normalized cycle counts $\talpha$ for \bluetext{$g'$-cycles, $g' < \ttg$, given by}  $\talpha = (\sum_{g'<\ttg}\frac{\alpha^{g'}_1}{T}, \sum_{g'<\ttg}\frac{\alpha^{g'}_2}{T}, \ldots, \sum_{g'<\ttg}\frac{\alpha^{g'}_n}{T})$ where $T = |\{4, 6, \ldots, \ttg-2\}|$. \lev{Once the CN $\tcn^{s}$ is selected, we update $\ttc^{g'}$, $\alpha^{g'}$, and $\talpha$ based on the new \bluetext{$g'$-cycles, $g' < \ttg$,} created by the edge between $\tcn^{s}$ and $v_j$  (Step 13 of Algorithm \ref{alg:EC-PEG})}.
Thus, for the case when $F = 1$, for each candidate CN $\tcn^{ca} \in \tK$, \bluetext{$g'$-cycles} are formed in the TG when an edge is established between $\tcn^{ca}$ and $v_j$. The cycles formed affect $\talpha$ differently for each $\tcn^{ca}$. In the ``entropy-constrained CN selection" procedure, we select the CN in $\tK$ that results in the minimum entropy $\tH(\talpha)$ for the resultant joint normalized cycle counts $\talpha$ after the addition of the edge. If there are multiple candidate CNs that result in the same minimum entropy, we select the CN with the least degree in the current TG setting and then randomly if multiple CNs have the least degree. Once we select the CN $\tcn^s$, we update the TG $\tGEC$ and also update $\ttc^{g'}$,  $\alpha^{g'}$ and $\talpha$ as mentioned before. If cycles formed are of length $g' \geq \ttg$, we again use the ``min degree CN selection" procedure. 
\greentext{The intuition behind the EC-PEG algorithm is that entropy is minimized when a distribution is concentrated, thus, selecting a CN that minimizes entropy leads to a much more concentrated distribution. Minimizing the entropy of the joint normalized cycle counts also ensures that the different cycle distributions are concentrated towards the same set of VNs.}

\vspace{-3pt}
\begin{remark}\label{remark:concentrationargument}
\greentext{When the $PEG()$ procedure in Algorithm \ref{alg:EC-PEG} returns $F = 1$ and cycle length $g' < \ttg$,} only \bluetext{the $g'$-cycles} are accounted for in the counts  $\ttc^{g'}$, $\alpha^{g'}$, and $\talpha$. It may so happen that along with \bluetext{the $g'$-cycles}, certain larger cycles of length $> g'$ are also formed when a new edge is added to the TG. In the EC-PEG algorithm, those cycles are not used for updating their respective cycle counts. Simulations show that it is sufficient to only consider the \bluetext{$g'$-cycles}, when we encounter $F=1$ with $g'$, in concentrating the final distributions $\td^{g'}$ (which includes all the \bluetext{$g'$-cycles}) as can be seen in Fig. \ref{fig:cycle_dist}. 
We call the cycles accounted for in the counts as ``considered" cycles. 
We conjecture that when considered \bluetext{$g'$-cycles} are formed in the EC-PEG algorithm, the non-considered cycles formed of length $> g'$ touch almost the same set of VNs as the considered \bluetext{$g'$-cycles}. \greentext{The EC-PEG algorithm concentrates all the considered $g'$-cycles for $g' < \ttg$ towards the same set of VNs by minimizing $\tH(\talpha$).
Thus, this procedure also concentrates the overall cycle distributions $\td^{g'}$ for $g' < \ttg$ towards the same set of VNs. 
It is however also possible to use the larger cycles to update the counts but this approach would incur higher complexity of finding all cycles of length $ < \ttg$ each time a new cycle is created.  }
\end{remark}

\vspace{-12.5pt}
\section{Simulation Results}\label{sec:sims}
\vspace{-2.5pt}
In this section, we present the performance of codes designed using the EC-PEG algorithm  and compare it with the performance of the original PEG algorithm \cite{PEG} when using the greedy sampling (GS) strategy described in Strategy \ref{remark:samplingstrategy}, and the performance achieved by \cite{CMT} using random LDPC codes and random sampling (RS).
For all codes designed using the EC-PEG algorithm, we choose $\ttg =10$, $d_v = 4$ and rate = $\frac{1}{2}$. 
Fig. \ref{fig:cycle_dist} shows the \bluetext{6-cycle} and \bluetext{8-cycle} distributions $\td^6$ and $\td^8$ generated using the EC-PEG and the original PEG algorithms for $n=100$ \deb{($g_{\min} = 6$ for both the codes)}.  The \lev{x-axis} in the figure are the VN indices $v_i$, arranged in the decreasing order of the \bluetext{6-cycle} fractions $\td^6_i$  and the \lev{y-axis} are the respective cycle fractions.
From Fig. \ref{fig:cycle_dist}, we can clearly see that the EC-PEG algorithm generates significantly concentrated distributions $\td^6$ and $\td^8$ compared to the original PEG algorithm and also the distributions are concentrated towards the same set of VNs. This is because the same VNs that have a high \bluetext{6-cycle} fraction also have a high \bluetext{8-cycle} fraction.
Fig. \ref{fig:S_g_avsN} shows the variation of the greedy 10-sample set $|S^a_{10}|$ for different codeword lengths. We see that the EC-PEG algorithm has a lower  $|S^a_{10}|$ than the original PEG algorithm which is a result of making the cycle distributions concentrated towards the same set of VNs. Since $\Sigma(4,10) \geq 17$, from Lemma \ref{lemma:zeroprob}, for all stopping sets of weight $\teta < 17$, $P_f(S^a_{10},\teta)$ is zero. Additionally, the green plot in Fig. \ref{fig:S_g_avsN} is the minimum number of samples $s$ required to have a $p_f(s,16) \leq 10^{-6}$ using random sampling as in \cite{CMT}. We see that to have $p_f(s,16) \leq 10^{-6}$, a large number of samples are needed compared to $|S^a_{10}|$.

Fig. \ref{fig:SS_dist} shows the stopping set distributions $ss^{\teta}$ for LDPC codes designed using the original PEG and the EC-PEG algorithms for $n=100$. The \lev{x-axis} in all the plots are the VN indices $v_i$, arranged in the decreasing order of the \bluetext{6-cycle} fractions $\td^6_i$. Clearly, for the EC-PEG algorithm the stopping set distributions are concentrated towards the same set of VNs as that of the \bluetext{6-cycle} distributions since the VNs towards the right (left) on the \lev{x-axis} have low (high) stopping set fraction.
This motivates the sampling choice of $\tS_{greedy}$ to minimize $p_f(s,\teta)$ according to Strategy \ref{remark:samplingstrategy}. 
\greentext{We note that a similar trend does not hold for the original PEG algorithm.} 
\begin{figure}[t]
    \centering
    \begin{subfigure}{0.5\linewidth}
\begin{minipage}{0.99\linewidth}
\begin{tikzpicture}
  \node (img) {\includegraphics[scale=0.15]{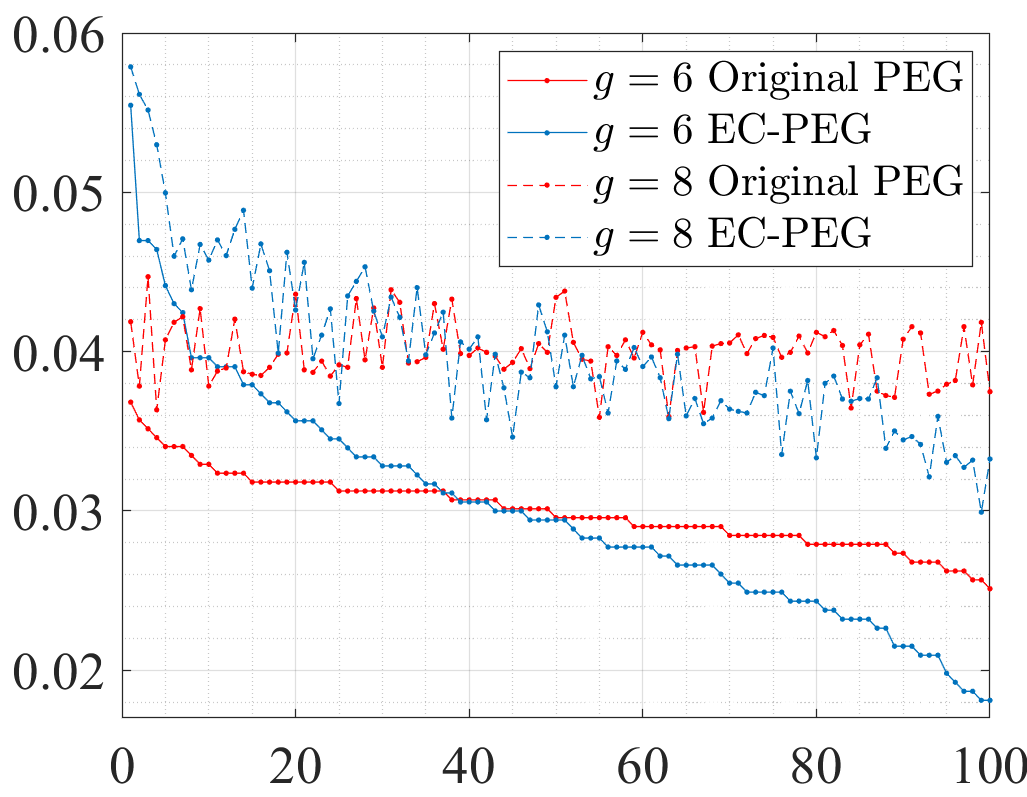}};
  \node[below=of img, node distance=0cm, yshift=1.1cm,font=\color{black}] {\footnotesize{VN index}};
  \node[left=of img, node distance=0cm, rotate=90, anchor=center,yshift=-1cm,font=\color{black}] {\footnotesize{$\td^g$}};
 \end{tikzpicture}
 \end{minipage}
 \vspace{-3pt}
\caption{}
\label{fig:cycle_dist}
    \end{subfigure}%
    \begin{subfigure}{0.5\linewidth}
\begin{minipage}{0.99\linewidth}
\begin{tikzpicture}
  \node (img) {\includegraphics[scale=0.15]{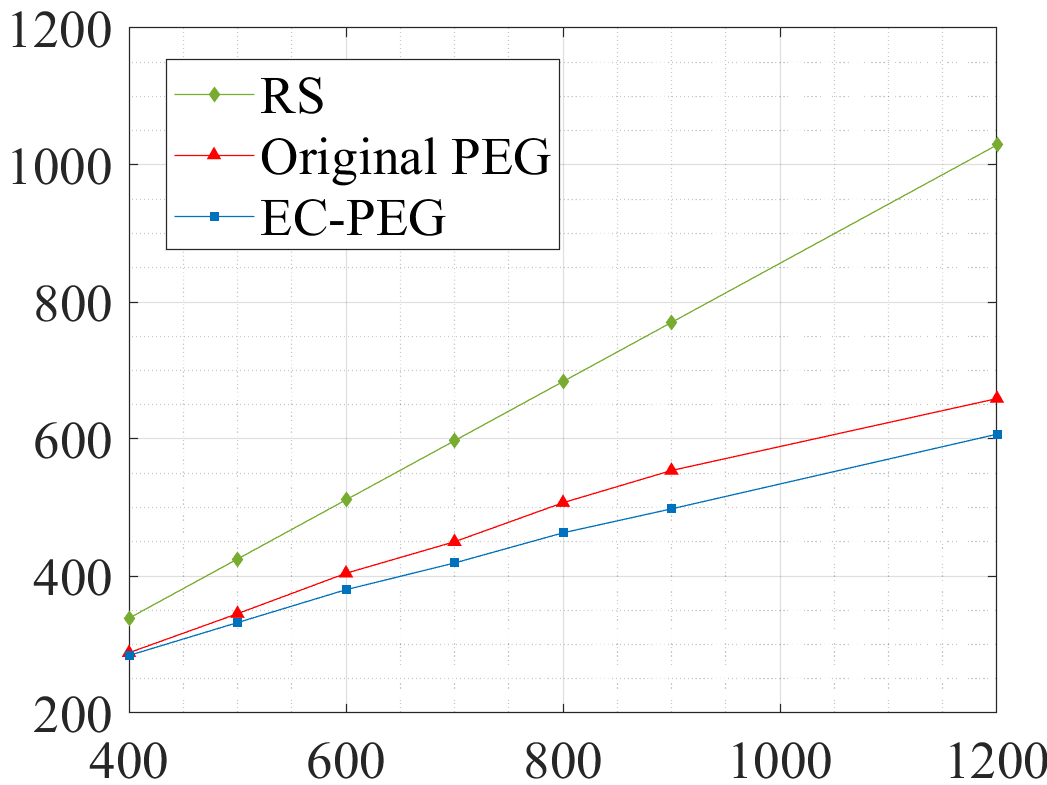}};
  \node[below=of img, node distance=0cm, yshift=1.1cm,font=\color{black}] {\footnotesize{No. of VNs $n$}};
  \node[left=of img, node distance=0cm, rotate=90, anchor=center,yshift=-1cm,font=\color{black}] {\footnotesize{$|\tS^{a}_{10}|$}};
 \end{tikzpicture}
 \end{minipage}
 \vspace{-3pt}
\caption{}
\label{fig:S_g_avsN}
    \end{subfigure}
     \vspace{-5pt}
    \caption{Results for LDPC codes with rate = $\frac{1}{2}$, $d_v = 4$; (a) Cycle distribution using different PEG algorithms for $n=100$, (b) $|\tS^{a}_{10}|$ vs. the number of  VNs.}
    \label{}
     \vspace{-12pt}
\end{figure}
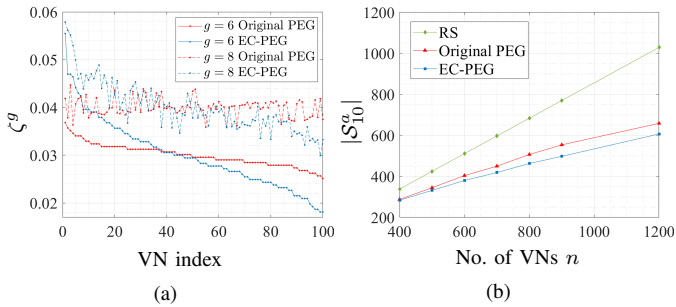

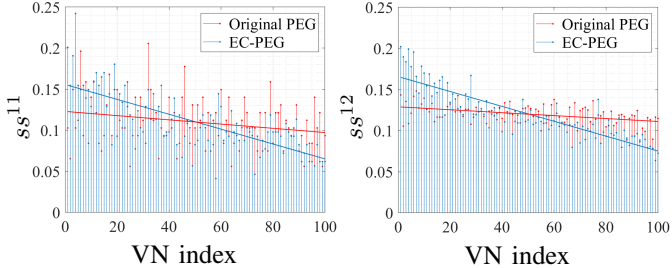
\begin{figure}[t]
    \centering
    \begin{subfigure}{0.5\linewidth}
\begin{minipage}{0.99\linewidth}
\begin{tikzpicture}
  \node (img) {\includegraphics[scale=0.15]{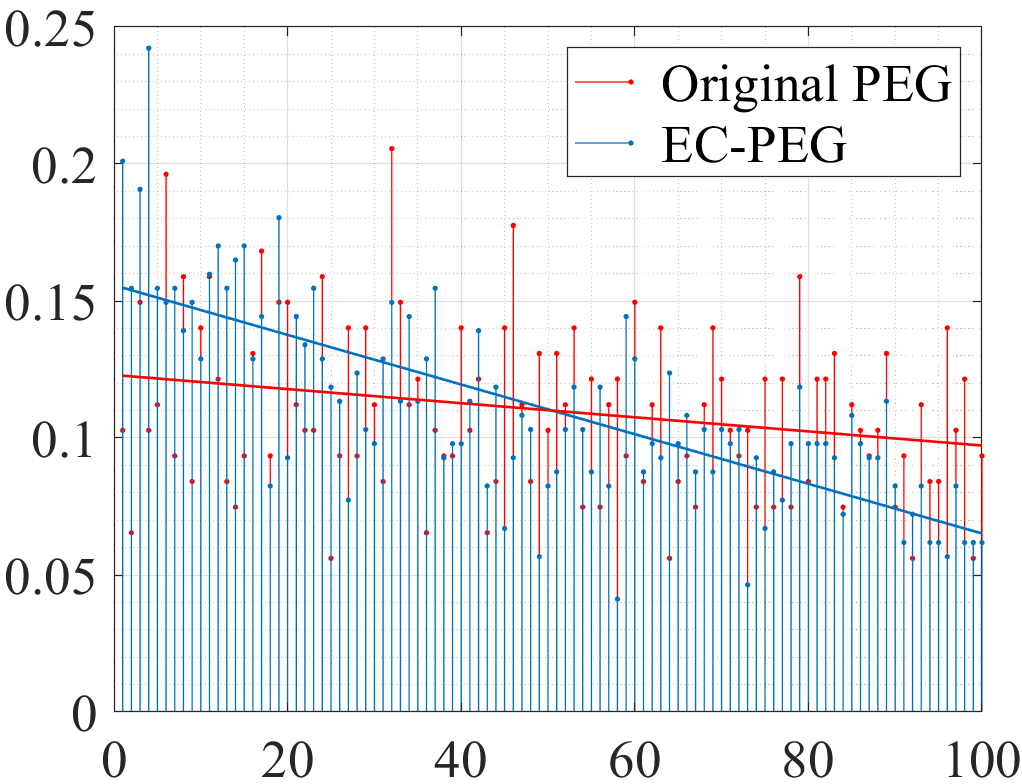}};
  \node[below=of img, node distance=0cm, yshift=1.1cm,font=\color{black}] {VN index};
  \node[left=of img, node distance=0cm, rotate=90, anchor=center,yshift=-0.95cm,font=\color{black}] {$ss^{11}$};
 \end{tikzpicture}
 \end{minipage}
    \end{subfigure}%
\begin{subfigure}{0.5\linewidth}
\begin{minipage}{0.99\linewidth}
\begin{tikzpicture}
  \node (img) {\includegraphics[scale=0.15]{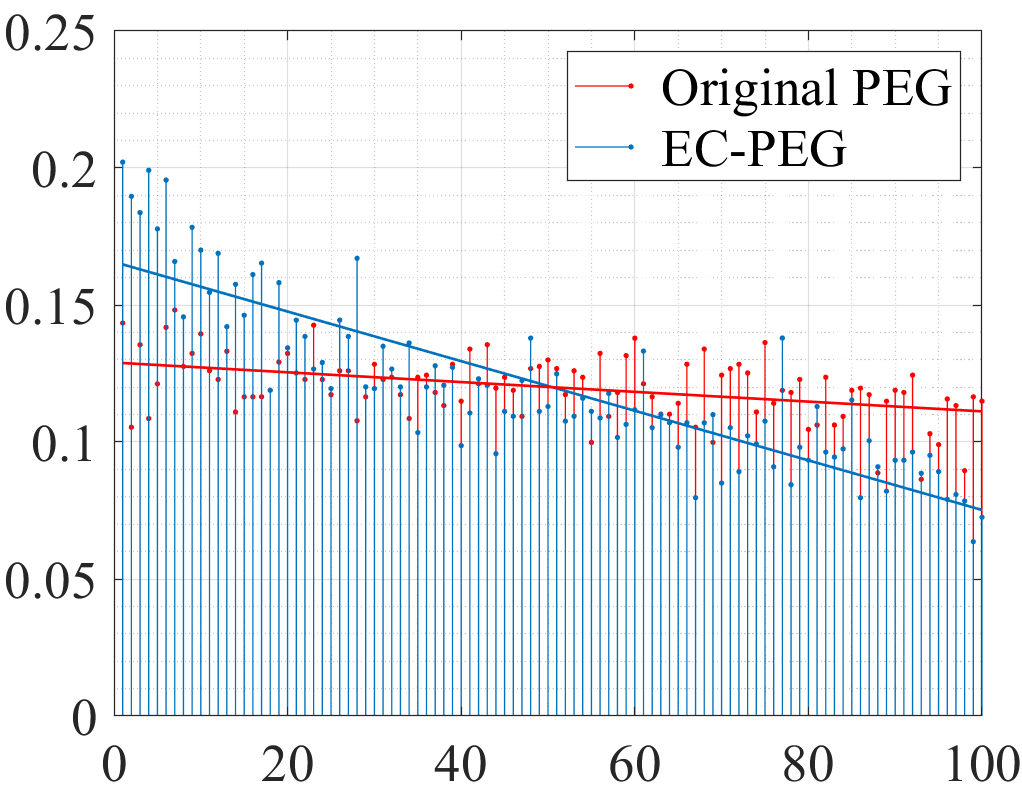}};
  \node[below=of img, node distance=0cm, yshift=1.1cm,font=\color{black}] {VN index};
  \node[left=of img, node distance=0cm, rotate=90, anchor=center,yshift=-0.95cm,font=\color{black}] {$ss^{12}$};
 \end{tikzpicture}
 \end{minipage}
\end{subfigure}
     \vspace{-4pt}
    \caption{Stopping Set distributions $ss^{\teta}$ for LDPC code with $n=100$, rate = $\frac{1}{2}$, $d_v = 4$; left: $\teta = 11$; right: $\teta = 12$. The lines are the best fit slopes for $ss^{\teta}$  indicating the graph slope.}
    \label{fig:SS_dist}
     \vspace{-14pt}
\end{figure}

\begin{figure}[t]
    \centering
\begin{minipage}{0.99\linewidth}
\begin{tikzpicture}
  \node (img)
  {\includegraphics[scale=0.27]{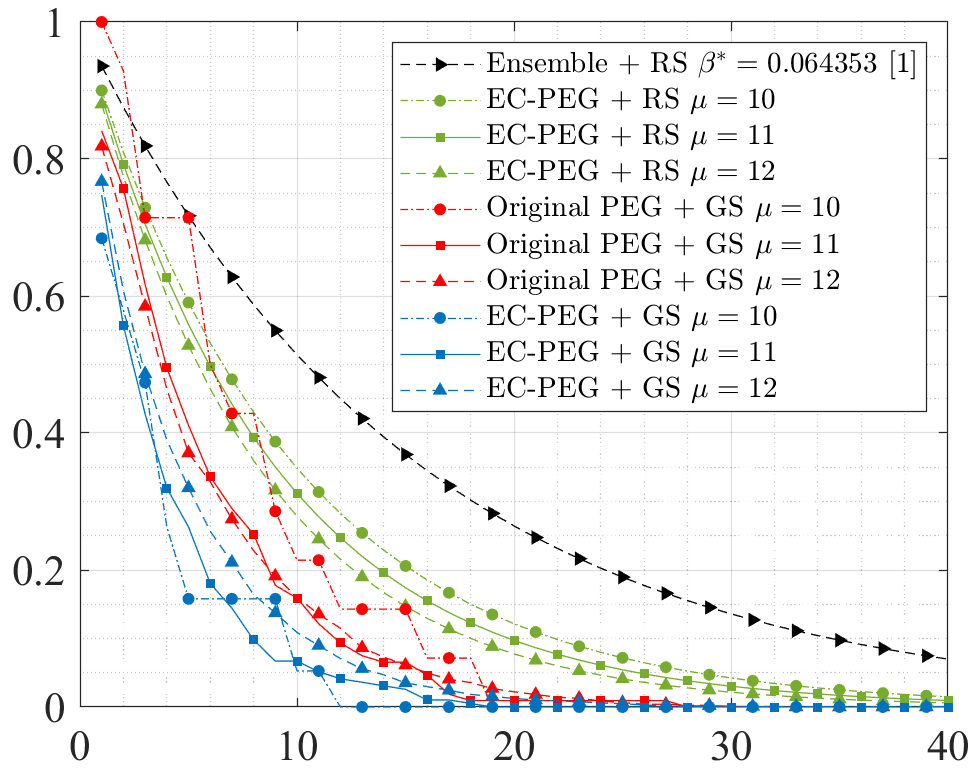}};
  \node[below=of img, node distance=0cm, yshift=1.1cm,font=\color{black}] {$s$};
  \node[left=of img, node distance=0cm, rotate=90, anchor=center,yshift=-0.9cm,font=\color{black}] {$p_f(s,\teta)$};
 \end{tikzpicture}
 \end{minipage}
 \vspace{-4pt}
\caption{$p_f(s,\teta)$ for various coding schemes and sampling strategies for $n=100$ and rate $= \frac{1}{2}$.}
\label{fig:P_f_s}
     \vspace{-12pt}
\end{figure}

\lev{In Fig. \ref{fig:P_f_s}, the probability of failure $p_f(s,\teta)$ is shown for $\teta = 10, 11$, and $12$ for varying codes with $n = 100$. The top black line is the $p_f(s,\teta)$ achieved using random sampling and a random LDPC code as in \cite{CMT} with stopping ratio $\beta^* = 0.064353$. The value $\beta^*$ is the best stopping ratio obtained  for a rate $\frac{1}{2}$ code following the method in \cite[section 5.3]{CMT} using parameters $(c,d) = (8,16)$. The green lines represent $p_f(s,\teta)$ obtained using random sampling for different $\teta$ calculated as $p_f(s,\teta) = (1 - \frac{\teta}{n})^s$ \bluetext{(when the malicious node randomly hides a stopping set of size $\teta$, $p_f(s,\teta)$ is still provided by this expression)}. The red and blue curves demonstrate the use of the greedy sampling strategy for codes designed by the original PEG and the EC-PEG algorithms, respectively, and $p_f(s,\teta)$ is calculated according to the equation provided in Strategy \ref{remark:samplingstrategy}. Note that the codes designed by the EC-PEG and the original PEG algorithms  have a minimum stopping set size of $9$, but $p_f(s,9) = 0, s \geq 3$ so we ignore these stopping sets. \redtext{Also, the maximum CN degree $d$ for the EC-PEG LDPC code is $d = 11$. Thus the incorrect coding proof size, which depends on $d$, is smaller than for the random construction with maximum CN degree $16$.}  Fig. \ref{fig:P_f_s} demonstrates the 3 major improvements of our design. 
The first benefit is the use of deterministic constructions to design finite-length codes that provide better guarantees for the stopping set ratio than random ensembles, as can be seen when comparing the black and green curves. 
The second benefit comes from using a greedy sampling strategy as opposed to random sampling, which can be observed by the significant reduction in $p_f(s,\teta)$ between the green and red curves. Finally, $p_f(s,\teta)$ is reduced further when the greedy sampling strategy is used on a concentrated LDPC code, as shown by comparing the red and blue curves. All of these benefits combine to significantly reduce $p_f(s,\teta)$ for a fixed sample size, demonstrating the efficacy of our techniques.}

\vspace{-3pt}
\section{Conclusion}\label{sec:conclusion}
\vspace{-2pt}
\redtext{ In this paper, we proposed a modification to the PEG algorithm to design LDPC codes that have a concentrated stopping set distribution.
We showed that our novel LDPC code design coupled with a greedy sampling strategy results in a much lower probability of failure compared to previous schemes. Our deterministic code construction allows for flexibility in code parameters including block length while providing good performance for this application.
Ongoing research is focused on stronger adversary models where malicious nodes are aware of the greedy sampling strategy. Current research is showing that it is possible to use concentrated LDPC codes with a different sampling strategy to account for such strong adversaries and is the subject of future work. 

}

\vspace{-0.1cm}
\section*{Acknowledgment}
\vspace{-0.1cm}
The authors acknowledge the Guru Krupa Foundation and NSF-BSF grant no. 2008728 to conduct this research work.

\newpage
 \appendices

\section{\bluetext{Algorithm for forming $g$-Sample sets}}\label{appendix:alg:g-sample}
\vspace{-0.2cm} 
\begin{algorithm}
\caption{Forming $g$-sample sets $S^a_g$, $g_{\min} < g \leq g_c$}\label{alg:g-sample}
\begin{algorithmic}[1]
\State \textbf{Inputs:} TG $\mathcal{G}$, $\ttg$, $g_{\min}$
\For{$g=g_{\min}$ to $\ttg - 2$, $g =$ even}
\State $\widehat{\mathcal{G}} = \mathcal{G}$
\If{$g=g_{\min}$} $\tC_{g} = \emptyset$
\Else{ $\tC_{g} = \cup_{g_{\min} \leq g' < g}\; \tC_{g'}$}
\EndIf
\State Purge all VNs in $\tC_{g}$ from $\widehat{\mathcal{G}}$
\While{$\widehat{\mathcal{G}}$ has $g$-cycles}
\State $\mathcal{V}_s$ = set of VNs that touch maximum number of\\\hspace*{1.2cm} $g$-cycles in $\widehat{\mathcal{G}}$
\State $v$ = VN selected from $\mathcal{V}_s$ uniformly at random
\State $\tC_{g} = \tC_{g} \cup \{v\}$
\State Purge $v$ from $\widehat{\mathcal{G}}$ 
\EndWhile
\EndFor
\State $S^a_g = \tC_{g-2}$, $g_{\min} < g \leq \ttg $
\State \textbf{Outputs:} $S^a_g$,  $g_{\min} < g \leq \ttg $
\State Auxiliary Outputs: $\tC_g$, $g_{\min} \leq g < \ttg $
\end{algorithmic}
\end{algorithm}
 
  \section{Proof of Lemma \ref{lemma:average_pf}}
  \label{appendix:lemma:average_pf}

Let $A(s) = \{A_1, A_2, \ldots, A_{|A(s)|}\}$ be the set of all possible subsets of size $s$ of the set of VNs $\{v_1, v_2, \ldots, v_n\}$. Also, let $\mathcal{SS}_{\teta} = \{S_1, S_2, \ldots, S_{|\mathcal{SS}_{\teta}|}\}$, where $S_i$ is a stopping set of weight $\teta$. Define the one-zero adjacency matrix $L$ of size $|\mathcal{SS}_{\teta}|$ $\times$ $|A(s)|$, where $L_{ij} = 1$ if and only if none of the VNs in $A_j$ touch $S_i$, else $L_{ij} = 0$. Let a general sampling strategy be denoted by $\mathrm{x} = \{x_1, x_2, \ldots, x_{|A(s)|}\}$, where
$x_j = \mathrm{Prob}$(VNs in set $A_j$ are sampled). Thus, if the malicious node picks a stopping set from $\mathcal{SS}_{\teta}$ randomly, the probability of failure to catch the malicious node is $\frac{1}{|\mathcal{SS}_{\teta}|}||L\mathrm{x}||_1$, where $||a||_1$ is the one norm of the vector $a$. The probability of failure $p_f(s,\teta)$ will be lower bounded by the optimal solution of the following optimization problem:

\begin{equation*}
\begin{aligned}
& \underset{\mathrm{x}}{\text{minimize}}
& & \frac{1}{|\mathcal{SS}_{\teta}|}||L\mathrm{x}||_1 \\
& \text{subject to}
& & 1^{T}x = 1,\\
& & & x \geq 0.
\end{aligned}
\end{equation*}

Let $k^* = \mathrm{argmin}_j||l_{j}||_1$, where $l_j$ is the $j^{th}$ column of $L$.  It is easy to see that the optimal solution of the above problem is $x^*_i = 1, i = k^*, x^*_i = 0, i\neq k^*$.

Recall that $\tau(\tS,\teta)$ is the fraction of stopping sets of weight $\teta$ touched by the sample set $\tS$. For any sampling strategy (with $s$ samples) $x_i = 1, i = k \in \{1, 2, \ldots, |A(s)|\}, x_i = 0, i\neq k$, i.e., sampling with probability one, the set $A_k$, $p_f(s,\teta) = \frac{1}{|\mathcal{SS}_{\teta}|}||l_{k}||_1 = 1 - \tau(A_k,\teta)$.
Thus, the optimal sampling strategy is always sampling a subset of VNs of size $s$ that touch the most number of stopping sets in $\mathcal{SS}_{\teta}$. The optimal solution becomes, $\frac{1}{|\mathcal{SS}_{\teta}|}||l_{k^*}||_1 = 1 - \tau(A_{k^*},\teta)$, which is the fraction of stopping sets not touched by VNs in $A_{k^*}$. This is exactly equal to $1 - \max_{\tS \subseteq \tV, |\tS| = s}\tau(\tS,\teta)$. Thus,

$$p_f(s,\teta) \geq \frac{1}{|\mathcal{SS}_{\teta}|}||l_{k^*}||_1 = 1 - \max_{\tS \subseteq \tV, |\tS| = s}\tau(\tS,\teta).$$

This lower bound is also achieved using the sampling strategy of $x^*_i = 1, i = k^*, x^*_i = 0, i\neq k^*$, i.e., sampling a subset of VNs of size $s$ that touch the most number of stopping sets in $\mathcal{SS}_{\teta}$.

\end{document}